\documentclass{IEEEtran}
\usepackage{textcomp}
\def\BibTeX{{\rm B\kern-.05em{\sc i\kern-.025em b}\kern-.08em
    T\kern-.1667em\lower.7ex\hbox{E}\kern-.125emX}}

\usepackage{amsmath, amssymb, amsfonts,  bm}
\usepackage{amsthm} 
\usepackage{mathrsfs}
\usepackage[mathcal]{euscript}
\usepackage{enumerate}
\usepackage{epsfig}
\usepackage{float}
\usepackage{color}
\usepackage{dsfont}
\usepackage{hyperref}
\usepackage{cite} 
\usepackage{algorithm} 
\usepackage{algorithmic} 
\usepackage{tikz-cd}
\usepackage{mathtools}
\usepackage{graphicx}
\usepackage{subcaption} 

\newtheorem{theorem}{Theorem}
\newtheorem{definition}{Definition}
\newtheorem{lemma}{Lemma}
\newtheorem{remark}{Remark}
\newtheorem{assumption}{Assumption}

\DeclareFontFamily{OT1}{pzc}{}
\DeclareFontShape{OT1}{pzc}{m}{it}{<-> s * [1.200] pzcmi7t}{}
\DeclareMathAlphabet{\mathpzc}{OT1}{pzc}{m}{it}

\newcommand*\mcapinn[2]{\vcenter{\hbox{$\mathsurround=0pt
  \ifx\displaystyle#1\textstyle\else#1\fi\bigcap$}}}

\newcommand*\mcupinn[2]{\vcenter{\hbox{$\mathsurround=0pt
  \ifx\displaystyle#1\textstyle\else#1\fi\bigcup$}}}

\def\begequarr{\begin{eqnarray}}
\def\endequarr{\end{eqnarray}}
\def\begequarrs{\begin{eqnarray*}}
\def\endequarrs{\end{eqnarray*}}
\def\begequ{\begin{equation}}
\def\endequ{\end{equation}}
\def\begequs{\begin{equation*}}
\def\endequs{\end{equation*}}
\def\begite{\begin{itemize}}
\def\endite{\end{itemize}}

\def\begcen{\begin{center}}
\def\endcen{\end{center}}
\def\begrem{\begin{remark}\rm}
\def\endrem{\end{remark}}
\def\ba{\begin{array}}
\def\ea{\end{array}}

\def\diag{\textnormal{diag}}

\newcommand{\vb}{ {\mathbf{v}}}

\newcommand{\bb}{\mathbf{b}}

\newcommand{\xb}{\mathbf{x}}
\newcommand{\yb}{\mathbf{y}}
\newcommand{\zb}{\mathbf{z}}

\newcommand{\Hb}{\mathbf{H}}

\newcommand{\Ib}{\mathbf{I}}

\newcommand{\Ab}{\mathbf{A}}

\newcommand{\Cb}{\mathbf{C}}

\newcommand{\eb}{\mathbf{e}}

\newcommand{\Sb}{\mathbf{S}}

\newcommand{\psib}{\bm{\psi}}
\newcommand{\zetab}{\bm{\zeta}}

\def\beeq#1{\begin{equation}{#1}\end{equation}}

\begin{document}

\title{Distributed Solvers for Network Linear
Equations \\ with Scalarized Compression} 

\author{\IEEEauthorblockN{Lei Wang, Zihao Ren, Deming Yuan, Guodong Shi}

 \thanks{This work  was supported in
part by Zhejiang Provincial Natural Science Foundation of China under Grant No. LZ23F030008 and in part by the National Natural Science Foundation of China under Grant No. 62203386.  }
 \thanks{L. Wang and Z. Ren are with the College of Control Science and Engineering, Zhejiang University, China (Email: lei.wangzju; zhren2000@zju.edu.cn). D. Yuan is with the School of Automation, Nanjing University of Science and Technology,  China (Email: dmyuan1012@gmail.com). G. Shi is with Australia Centre for Robotics, The University of Sydney, Australia (Email: guodong.shi@sydney.edu.au). Corresponding author: D. Yuan.}

}

\maketitle

\begin{abstract}                          
 Distributed computing is fundamental to multi-agent systems, with solving distributed  linear equations as a typical example. In this paper, we study distributed solvers for network linear equations over a network   with  node-to-node  communication messages compressed as scalar values.  {Our key idea lies in a dimension compression scheme that includes a dimension-compressing vector and a data unfolding step. The compression vector applies to individual node states as an inner product to generate a real-valued message for node communication. In the unfolding step, such scalar message is then plotted along the subspace generated by the compression vector for the local computations.} We first present a compressed consensus flow that relies only on such scalarized communication, and show that linear convergence can be achieved with well excited signals for the compression vector. We then employ such a compressed consensus flow as a fundamental consensus subroutine to develop distributed continuous-time and discrete-time solvers for network linear equations, and prove their linear convergence properties under scalar node communications. With scalar communications, a direct benefit would be the reduced node-to-node communication channel  {burden} for distributed computing. Numerical examples are presented to illustrate the effectiveness of the established theoretical results.
\end{abstract}

\begin{IEEEkeywords}                          
Network linear equation; Consensus;  Compression; Scalar communication
\end{IEEEkeywords}

\vspace{-1em}
\section{Introduction}
The rapid development in large-scale intelligent systems such as smart grid, intelligent transportation, and cyber-physical systems is increasingly relying on distributed solutions to sensing, estimation, control, and signal processing of multi-agent systems \cite{magnusbook,martinez07,kar12,Rabbat2010}. Such systems are organized as interconnected networks where nodes represent agents and subsystems located at physically decentralized locations, and links represent information, data, or energy flows among the agents and subsystems. In the network, distributed computations can be achieved with agents only communicating (or interacting) with a small number of neighbors \cite{magnusbook}. 
The promise of such technology advances has been radical  improvements in scalability and efficiency for the operation of complex dynamical systems. In this paper, we are mainly interested in a particular distributed computation problem, where each agent holds a linear algebraic equation and aims to compute the solution of the network linear algebraic equations in a distributed manner, with  a wide range of engineering applications (see e.g. \cite{Pas(2012)} for power system monitoring  and \cite{Pan(2021)} for versatile LiDAR SLAM).

In the literature, there have been many investigations on  distributed solutions to  network linear equations, with the majority established on the  consensus seeking algorithm \cite{magnusbook}.
By employing the  consensus algorithm as a subroutine for distributed information aggregation, 
a number of distributed network linear equation algorithms have been established regarding the convergence to the exact/least-squares solution and convergence rates for both discrete-time and continuous-time settings under various network conditions \cite{lu09-1,lu09-2,liu13,mou13,asuman14,jadbabaie15,Morse-TAC15,hedge19,Yi(2023)}, even at finite time steps \cite{tao2020}. 
Moreover, the connection between distributed discrete algorithms and dynamical flows has also been explored \cite{elia,cotes14,brian15,shi2017networkTAC}, with the idea of using differential equations to study recursive algorithms traced back to \cite{ljung1977analysis}. 
See \cite{Wang_survey(2019)} for a thorough review of recent advances on distributed continuous-time and discrete-time solvers for network linear equations.

In these distributed computation algorithms, node communication in general must be of the same dimension as the network decision variable which may be of high dimension. Furthermore, most of the algorithms experience slower convergence as the size of the network increases, especially for networks with sparse communication links. Therefore, the communication overheads in distributed algorithms have been an important issue. 
To handle such issues, tremendous efforts have been devoted with several communication compression or quantization strategies developed for different distributed computation algorithms without critically upsetting convergence rates, such as unbiased compressor \cite{liu2021linear,Doan(2020)}, the biased  compressor \cite{Kolo(2019),Beznosikov2020Onbiased}, the standard uniform quantizer \cite{Nedic(2008),li2020}, etc.  {
Besides, encoding-transmission-decoding processes are developed in \cite{lei2020,liu2020} for  reducing the number of transmitted bits.}
Readers of interest are referred to \cite{Yi2023Communication} for a thorough survey on the developments of communication compression in distributed computation.

In this paper, we study distributed computation algorithms over networks for the exact solution of network linear equations,  where node-to-node communication messages are compressed to scalars. In the proposed flows, a dimension compressing vector is first applied to individual node states to generate a real-valued  message for node communication, and then in the local computations, the scalar message is unfolded along the subspace generated by the compressing vector.  First of all, a compressed consensus flow with directly applying the scalarized compression scheme is presented, with a rigorous proof of the necessary and sufficient condition on the compression vector for a linear convergence  in view of the classical uniform complete observability theory. With such a compressed  consensus algorithm as a basic subroutine for information aggregation,  a distributed network linear equation flow with scalarized communications and  its discrete-time version are then proposed. It is shown that with the compression vector appropriately designed to satisfy a persistent excitation condition, the proposed both continuous-time and discrete-time solvers can solve the network linear equations  {with linear convergence for an appropriately designed stepsize}. The main contribution of this paper lies in developing a new scalarized communication compression strategy and  proving that the resulting distributed solvers by directly incorporating such a compression strategy can solve the network linear equations  {with linear convergence}. Numerical examples are presented to show that the node-to-node communication burden  {to complete the computation problem} can be reduced by using the proposed scalarized compression strategy.

The idea of compressing communication was addressed in \cite{Kolo(2019)} for average consensus and network convex optimization, where the communication is compressed by employing a sparsification operator, to reduce the number of bits required for representing the communication information. In \cite{Kolo(2019)} the communication dimension is still the same as that of the node states, leading to higher communication channel  {burden} as the state dimension increases. This is different from our proposed compression approach, where the dimension of the messages being communicated is compressed into   {one scalar}, independent of the node state dimension.
 {The proposed solvers are also relevant to \cite{wang2018}, where two scalar states are communicated, while the 
 equation is disassembled in another way, requiring more nodes.}

 {
\emph{\bf Notation.}
In this paper, $\|\cdot \|$
denotes $2$-norm, i.e., $\|x\|=\sqrt{x^\top x}$ for $x\in\mathbb{R}^n$ and $\|A\|=\max_{\|x\|=1}\|Ax\|$ for $A\in\mathbb{R}^{n\times m}$. The notation $\mathbf{1}_n(\mathbf{0}_n)$, $\mathbf{I}_n$ and $\{{\bf e}_1,...,{\bf e}_m\}$ denote the column one (zero) vector, identity matrix and base vectors in $\mathbb{R}^d$, respectively. The symbol $\otimes$ denotes the Kronecker product. $A\succ B$ ($A{\succeq}B$) means $A-B$ is positive (semi)definite for any two compatible square matrices $A,B$. 
}
\vspace{-1em}
\section{Problem Definition}\label{sec2}

In this paper, we are mainly interested in solving network linear algebraic equation, i.e., each agent $i$ holds a  linear algebraic equation $
\mathcal{E}_i: \ \mathbf{H}_i^\top\mathbf{v}=b_i$
with an unknown $\mathbf{v} \in\mathbb{R}^m$, $\mathbf{H}_i \in\mathbb{R}^{m}$ and $b_i\in \mathbb{R}$.
The agents aim to solve the network linear algebraic  equation
\begin{equation}\label{eq:LAE}
\mathcal{E}:  \ \mathbf{H} \mathbf{v} =\mathbf{b}
\end{equation}
with  $\mathbf{H}= [\mathbf{H}_1 \,\dots  \, \mathbf{H}_n]^\top\in\mathbb{R}^{n\times m}$ and $\mathbf{b}=[b_1\,\dots \,b_n]^\top\in\mathbb{R}^{n}$.
For the solvability of $\mathcal{E}$, we impose the assumption below.

\begin{assumption}\label{ass-sol}
    For the linear equation $\mathcal{E}$ in \eqref{eq:LAE}, there holds $$
{\rm rank}(\mathbf{H})={\rm rank}\big([\mathbf{H}\ \mathbf{b}]\big)=m.
$$
\end{assumption} 
The above assumption implies that the network linear equation $\mathcal{E}$ has the exact unique solution $\vb^\ast=(\Hb^\top\Hb)^{-1}\Hb^\top\bb$.

 {For such a distributed computation problem \eqref{eq:LAE}, the  conventional algorithms (e.g. in \cite{shi2017networkTAC,Morse-TAC15,tao2020}) are established on the consensus algorithm.} 
During distributed computing processes, each node $i$ holds a dynamical state $\mathbf{x}_i (t)\in \mathbb{R}^m$ at time $t$,  typically of the same  dimension as the decision variable of the network-wise problem. Then these $\mathbf{x}_i(t)$ are shared over the network among nodes' neighborhoods, specified by a simple,  {  undirected, and connected}   graph $\mathrm{G}=(\mathrm{V},\mathrm{E})$. Distributed  algorithms  are  recursions of the $\mathbf{x}_i(t)$, which  depend  only on  the local dataset at node $i$  and the states that node $i$  {received}.

In a successful algorithm design,   each $\mathbf{x}_i(t)$ converges to the solution of the global problem \cite{shi2017networkTAC}.
However, when the dimension of the decision variable is high, the communication  {burden} along each node-to-node communication channel for the computation process becomes significant, as each node is sending a vector of the same dimension to all its neighbors at any given time instance.
Therefore, it is desirable for the nodes not to share the full $\mathbf{x}_i(t)$ but rather certain  scalar-valued messages, and still facilitate the distributed computation tasks. We are interested in the following problem.


\noindent{\bf Problem.}  Each node  $i$ holds  $\mathbf{x}_i (t)\in \mathbb{R}^m$ at time $t$;  shares with its neighbors  $\mathbf{y}_i(t)=\mathbf{C}^\top(t) \mathbf{x}_i(t) \in \mathbb{R}$ where the dimension compression vector $\mathbf{C}(t)\in\mathbb{R}^m$ is piecewise continuous and bounded.  Design distributed algorithms that drive each $\mathbf{x}_i(t)$ to the solutions of the  global network-wise problem $\mathcal{E}$.

Before the close of this section, we let $[a_{ij}]\in\mathbb{R}^{n\times n}$ be a weight matrix complying with the graph $\mathrm{G}$, where $a_{ij}\geq 0$ for all $i$ and $j$, and $a_{ij}>0$ if and only if $(i,j)\in\mathrm{E}$.
 Denote $\mathbf{L}_{\mathrm{G}}$ as the Laplacian of the graph $\mathrm{G}$, where $[\mathbf{L}_{\mathrm{G}}]_{ij}=-a_{ij}$ for all   $i\neq j$, and $[\mathbf{L}_{\mathrm{G}}]_{ii}=\sum_{j=1}^n a_{ij}$ for all $i\in\mathrm{V}$. As $\mathrm{G}$ is { assumed to be connected and undirected}, the eigenvalues of $\mathbf{L}_{\mathrm{G}}$, denoted by $\lambda_i$, $i=1,\ldots,n$  take the form $0=\lambda_1<\lambda_2\leq\cdots\leq\lambda_n$ by \cite{magnusbook}. 
Let $\mathbf{C}(t) \in \mathbb{R}^m$ be a piece-wise continuous dimension compression vector for the decision space, and there is no loss of generality to let $\mathbf{C}(t)$ be \emph{normalized}, i.e., $\|\Cb(t)\|=1$ for all $t$ throughout the paper. 

\vspace{-1em}
\section{Consensus with Scalar Communication}
\label{sec3}

Consensus is a fundamental algorithm that acts as a subroutine in numerous distributed computation problems, including network linear equations. 
In this section, we propose a distributed compressed consensus flow over the  graph $\mathrm{G}$, where node-to-node information exchanges are real numbers. We show that along the flow node states will reach consensus at an exponential rate under a persistent excitation (PE) condition on the dimension compression vector.

 We propose the following consensus flow
\begin{equation}
  \label{eq:CAC}
\dot{\mathbf{x}}_i = \displaystyle\sum_{j=1}^n a_{ij}    \mathbf{C}(t)  \big(\mathbf{y}_j- \mathbf{y}_i\big)\,,\
\mathbf{y}_i=\mathbf{C}^\top(t)  \mathbf{x}_i
  \,,\ i\in\mathrm{V}\,.
\end{equation}

Clearly for the implementation of the flow, each node $i$ only needs to send a scalar   $\mathbf{y}_i(t)$ to its neighbors at each time $t$. The closed-loop  flow can be written as
\begin{equation}
  \ba{rcl}\label{ccf}
\dot{\mathbf{x}}_i(t) = \sum\limits_{j=1}^n a_{ij}   \mathbf{C}(t)   \mathbf{C}^\top(t) (\mathbf{x}_j(t)- \mathbf{x}_i(t)), \ i\in\mathrm{V}
  \ea
\end{equation}
which is termed  a {\em compressed consensus flow}. { For convenience, as in \cite{liu2021linear,Beznosikov2020Onbiased,Yi2023Communication} we define the \emph{scalarized compressor} as $\Cb_0(t, \xb):=\Cb(t)\Cb^\top(t)\xb\in\mathbb{R}^m$.}
Assume that the initial time of the flow is  $t_0\geq0$, and denote the network initial value $\mathbf{x}(t_0):=\big[\mathbf{x}_1^\top(t_0)\ \dots \ \mathbf{x}_n^\top(t_0)\big]^\top$ and the average $\mathbf{x}^{\star}=\sum_{k=1}^n \mathbf{x}_k(t_0)/n$.
We introduce the following definition.
\begin{definition}\footnote{Note that  linear convergence equals to exponential convergence by observing  $\gamma^{(t-t_0)}=\mbox{exp}(-\hat\gamma(t-t_0))$ with $\hat\gamma=\ln(1/\gamma)$.}
The flow (\ref{ccf}) achieves global linear consensus (GLC) if there are constants $c>0$ and $\gamma\in(0,1)$ such that, for all $\mathbf{x}(t_0)\in\mathbb{R}^{mn}$, 
\begin{equation}\label{eq:GEAC}
\big\| \mathbf{x}(t)- \mathbf{1}_{n}\otimes\mathbf{x}^{\star}\big\|^2\leq  c\|\mathbf{x}(t_0)\|^2 \gamma^{(t-t_0)}\,,\quad t\geq t_0.
\end{equation}
\end{definition}

We impose the following PE assumption on the dimension compression vector $\mathbf{C}(t) \in \mathbb{R}^m$.

\begin{assumption}\label{ass-PE}
The compression vector $\mathbf{C}$ is PE, i.e., there exist constants $\alpha,T>0$ such that
\begin{equation}
\int_{t}^{t+T} \mathbf{C}(s)  \mathbf{C} ^\top (s)ds  {\succeq} \alpha \mathbf{I}_m\,,\qquad \forall t\geq t_0.
\end{equation}
\end{assumption}

The above PE condition  inherits from the field of  adaptive control  or linear time-varying (LTV) systems  for globally exponential stability \cite{Brian-TAC-1977}. The intuition behind Assumption \ref{ass-PE} is that the vector $\mathbf{C} $ should sample the entire state space $\mathbb{R}^m$ evenly enough so that each dimension of $\mathbf{x}_i$ is covered by sufficient excitations. 
We present the following results.

\begin{theorem}\label{thm1}
The compressed consensus flow (\ref{ccf}) achieves GLC iff the compression vector $\mathbf{C}$ satisfies Assumption \ref{ass-PE}.
\end{theorem}


The proof of Theorem \ref{thm1} is established on the classical results for stability of LTV systems \cite{Brian-TAC-1977}.
The following result characterizes the corresponding convergence rate.

\begin{theorem}\label{thm2}
Let Assumption \ref{ass-PE} hold. Then the compressed consensus flow (\ref{ccf}) achieves GLC, i.e., satisfying \eqref{eq:GEAC}  with
\[
\ba{l}
\gamma=\Big(1-\frac{2\alpha\lambda_2}{(1+T\lambda_n)^2}\Big)^{1/T}\,,\quad
c={1}/\Big(1-\frac{2\alpha\lambda_2}{(1+T\lambda_n)^2}\Big)\,.
\ea
\]
\end{theorem}

\begin{remark}
The implementation of the algorithm (\ref{ccf}) requires to assign each agent with a common compression vector-valued function $\Cb(t)$ satisfying Assumption \ref{ass-PE}.
A simple effective strategy is to set a message transmission rule for all agents, i.e., each agent transmits only one entry of its state in order at each unit time interval $[k,k+1)$ for $k\in\mathbb{N}$. This in turn induces a trivial design of $\Cb(t)$ as $\Cb(t)=\eb_i$ with $i=1+ (k\mod m)$  for $t\in [k,k+1)$, $k\in\mathbb{N}$.  In addition to step functions, there are many other kinds of functions, e.g., trigonometric functions, that can be utilized to design the compression vector, e.g., $\Cb(t)=[\sin{t},\,\cos{t}]$.
\end{remark}

\vspace{-1em}
\section{Main Results}\label{sec5}
\label{sec4}
In this section, we apply the proposed compressed consensus algorithm to solve the  network linear  equation $\mathcal{E}$. 

\vspace{-1em}
\subsection{Continuous-Time Compressed Linear Equation Flow}

The compressed consensus flow \eqref{eq:CAC} can be combined with different distributed network linear equation solvers in \cite{elia,cotes14,brian15}. In this paper, we focus on 
the ``consensus + projection'' algorithm  \cite{shi2017networkTAC} and combine flow \eqref{eq:CAC} with it, where node communications are still kept as scalars
\begin{equation}
\ba{rcl}\label{eq:LAE_solver}
\dot{\mathbf{x}}_i &=& \displaystyle\sum_{j=1}^n a_{ij}   \mathbf{C}(t)  \big(\mathbf{y}_j- \mathbf{y}_i\big) - s\mathbf{H}_i(\mathbf{H}_i^\top\xb_i-b_i)\\
\mathbf{y}_i&=&\mathbf{C}^\top(t)  \mathbf{x}_i
\ea
\end{equation}
with a constant stepsize $s>0$, for $i\in\mathrm{V}$.  The flow \eqref{eq:LAE_solver} is fundamentally a direct incorporation of the scalarized compressor $\Cb_0(t,\xb):=\Cb(t)\Cb^\top(t)\xb$ into the ``consensus + projection'' algorithm  \cite{shi2017networkTAC}.

If we assume $\mathbf{y}_i(t)=\mathbf{x}_i(t)$ for all $i\in\mathrm{V}$ and let $s=1$ in (\ref{eq:LAE_solver}), the flow is reduced to the ``consensus + projection'' flow  studied in \cite{shi2017networkTAC}. Let $\rho_m=\frac{1}{n}\sigma_m(\Hb^\top\Hb)$ and $h_{M}=\max\{\|\Hb_1\|,\ldots,\|\Hb_n\|\}$. 
We are ready to present the following result on the convergence of the  flow (\ref{eq:LAE_solver}).
\begin{theorem}\label{thm-LAE-new}
Let Assumptions \ref{ass-sol} and \ref{ass-PE} hold.  {Then for $s>0$}, the flow (\ref{eq:LAE_solver}) solves the network linear equation $\mathcal{E}$  {with linear convergence}, i.e., 
\begin{equation}\label{eq:GEAC_LAE}
\| \mathbf{x}(t)-  \mathbf{1}_n\otimes\mathbf{v}^{\ast}\| =  \mathcal{O}\left({\gamma_{f}}^t\right),
\end{equation}
where the convergence rate 
\[
\gamma_f=\Big(1-\frac{2\bar\alpha}{(1+(\lambda_n+ 2sh_M^2)) T)^2}\Big)^{1/T}\,,
\]
with $\bar\alpha = ({\alpha'  - \sqrt{{\alpha'} ^2-4\lambda_2\alpha\rho_m Ts}})/{2}$ for $\alpha' = \lambda_2\alpha + (h_M^2+\rho_m)Ts$.
\end{theorem}

\vspace{-1em}
\subsection{Discrete-Time Compressed Linear Equation Dynamics}

The discrete-time version of Assumption \ref{ass-PE} is given below.

\begin{assumption}
\label{ass-PE-dis}
    There exist $\alpha_{\rm d}>0$ and $K\geq m$ such that
\begin{equation}
\sum_{s=k}^{k+K-1}  {\mathbf{C}[s]  \mathbf{C}[s]^\top}   {\succeq} \alpha_{\rm d} \mathbf{I}_m\,,\qquad \forall\, k\in\mathbb{N}.
\end{equation}
\end{assumption}

The discrete counterpart of the flow \eqref{eq:LAE_solver}, using Euler approximation, can be described as
\begin{equation}
    \label{eq:dt-DCO}
    \ba{rcl}
\mathbf{x}_i[k+1] &=& \xb_i[k] + h\sum\limits_{j\in\mathrm{V}} a_{ij}    {\mathbf{C}[k]}  \big(\mathbf{y}_j[k]
  - \mathbf{y}_i[k]\big) \\ && - s \mathbf{H}_i(\mathbf{H}_i^\top\xb_i[k]-b_i)\\
\mathbf{y}_i[k]&=&\mathbf{C}^\top[k]  \mathbf{x}_i[k].
\ea
\end{equation}
with constant stepsizes $\frac{2}{\lambda_n}>h>0$, and $s>0$. 

\begin{theorem}\label{thm-dt-DCO-new}
 Let Assumptions \ref{ass-sol} and \ref{ass-PE-dis} hold.    Then given any $\frac{2}{\lambda_n}>h>0$, there exists $s^\ast>0$ such that for all  $0<s<s^\ast$,  the discrete-time  algorithm (\ref{eq:dt-DCO}) solves the network linear equation $\mathcal{E}$  {with linear convergence, i.e., 
$$
\|\xb[k] - \mathbf{1}_n\otimes\vb^\ast\|^2 = \mathcal{O}(\gamma_d^k)\,,\quad \gamma_d\in(0,1)\,.
$$
}
\end{theorem}

\begin{remark}
 {
The compressed solvers (\ref{eq:LAE_solver}) and (\ref{eq:dt-DCO}) are established on the ``consensus + projection'' solver  \cite{shi2017networkTAC}, of which the convergence rate is determined by both quantities of consensus and projection. Intuitively, the consensus process becomes slower by adding  scalarized compressors. This thus definitely leads to a slower convergence of the compressed solvers, in contrast with the original one with no compression. Despite of this, we note that at each iteration the communication amount after using the scalarized compressor is reduced to only $\frac{1}{m}$ of the original ones. In other words, our compressed solvers (\ref{eq:LAE_solver}) and (\ref{eq:dt-DCO}) are beneficial of reducing the total communication burden of each link if their convergence rates are $m$ times faster than the original solvers, as addressed in the subsequent numerical simulations. 
}
\end{remark}

\begin{remark}
 { The implementation of the proposed scalarized compressor $\Cb_0(t,\xb):=\Cb(t)\Cb^\top(t)\xb$  is  comprised of compression and decompression processes. Specifically, for each agent, the vector $\Cb$ is used to compress the message vector into a scalar for transmission, and the received scalars are unfolded into a vector using the $\Cb$ again for decompression. This compressor $\Cb_0(\cdot)$, that can be directly combined with the ``consensus + projection'' algorithm  \cite{shi2017networkTAC}, is different from the existing ones in \cite{liu2021linear,Beznosikov2020Onbiased,Yi2023Communication}, where there is no such compression-decompression process and extra design is required to guarantee linear convergence. Moreover, due to the presence of time-varying compression vector, the convergence analysis  turns out challenging and complex, particularly for the discrete-time one, and the LTV control system tools are borrowed to construct Lyapunov functions (see Appendix E).
    }
\end{remark}

\vspace{-1em}
\section{Numerical Simulations}

In this section, numerical simulations are presented to verify the validality of the proposed algorithms. Consider a network of $n=10$ nodes over a
circle communication graph, where each edge is assigned with the same unit weight and each node holds a linear equation with some randomly generated $(\mathbf{H}_i,b_i)$, satisfying Assumption \ref{ass-sol} and having the unique solution  $\mathbf{v}^\ast= [2,1,3,4,-1]$.



We first investigate the coninuous-time solver  \eqref{eq:LAE_solver} with the compression vector  
$\mathbf{C}(t)= \mathbf{e}_i$ with $i=1 + (k \mod m) $  for $t\in [k{\Delta}t,[k+1]{\Delta}t)$, $k\in\mathbb{N}$ and ${\Delta}t=0.01$.
It is clear that  $\|\mathbf{C}(t)\|=1$ and Assumption \ref{ass-PE} is satisfied with $T=5{\Delta}t$ and $\alpha_d={\Delta}t$. 
By letting the stepsize constant $s=3,1,0.5$, we implement the proposed compressed  flow \eqref{eq:LAE_solver}, with simulation results  given, where the simulation results for the corresponding standard flows without introducing compression, i.e., $\Cb=\Ib_m$ are also given. It can be seen that both flows solve the linear equation  {with linear convergence}, and the proposed flow \eqref{eq:LAE_solver} shows a  {lower} convergence rate. Despite of this, given the desired computation accuracy $\|\mathbf{x}-1_n\otimes\mathbf{v}^\ast\|/n=10^{-2}$, we note that the time required by the compressed one \eqref{eq:LAE_solver} is strictly  less than $m=5$ times of the standard flow under all these cases. This, together with the fact that  the communication burden of the standard flow is $m=5$ times larger than the compressed one \eqref{eq:LAE_solver} at each round, implies that the proposed compressed  flow \eqref{eq:LAE_solver} has the benefit of reducing the total communication burden to solve the  computation task.

We then verify the effectiveness of the discrete-time compressed solver \eqref{eq:dt-DCO}, implemented with $h=0.2$ and $ {\mathbf{C}[k]}=\eb_i$ with $i= 1 + (k \mod m)$ for $k\in\mathbb{N}$. It is noted that $ {\mathbf{C}[k]}$ satisfies  Assumption \ref{ass-PE-dis} with $K=5$ and $\alpha_d=1$. 
By letting the stepsize constant $s=0.02,0.002,0.0005$, the corresponding simulation results are  shown. Given the desired accuracy level $10^{-2}$, under all these three cases the iteration steps required  by the compressed algorithm \eqref{eq:dt-DCO} is strictly  less than $m=5$ times of the standard algorithm, verifying the advantage of  \eqref{eq:dt-DCO} similarly in reducing the communication burden.

{ 
In the literature, there are also other types of compressors, such as Unbiased $l$-bits quantizer $\Cb_1$ \cite{liu2021linear}, Greedy (Top-k) sparsifier $\Cb_2$ \cite{Beznosikov2020Onbiased} and Standard uniform quantizer $\Cb_3$ \cite{Yi2023Communication} 
\[
\ba{c}
{\Cb}_1[\xb]=\frac{\|\xb\|_{\infty}}{2^{l-1}}{\rm sign}(\xb)*\lfloor\frac{2^{l-1}|\xb|}{\|\xb\|_{\infty}}+\overline{\omega}\rfloor,\\
{\Cb}_2[\xb]=\sum_{s=1}^{k}[\xb]_{i_s}{\rm e}_{i_s},\quad
{\Cb}_3[\xb]=\lfloor \xb+\frac{\mathbf{1}_m}{2}\rfloor\,,
\ea
\]
where $*$ denotes Hadamard product. 
For a fair comparison, we directly incorporate these compressors into the ``consensus + projection" algorithm \cite{shi2017networkTAC} to derive the corresponding compressed solvers, by using ${\Cb}_j$, $j=1,2,3$ to replace the scalarized compressor $\Cb_0$, that is $\Cb[k]\yb_i[k]$ in \eqref{eq:dt-DCO}. The resulting simulation results are presented. It can be seen that  the proposed  algorithm \eqref{eq:dt-DCO} has the advantages of maintaining linear convergence to the  solution. On the other hand, it is worth noting that some levels of computation residuals can be observed when directly applying these compressors $\Cb_1, \Cb_2, \Cb_3$, and one may remove such residuals, but with some further modifications to the algorithm, e.g., replacing the transmitted messages by state estimation errors through adding extra filter dynamics \cite{Yi2023Communication}, though it is out of scope of this paper. 
}

\section{Conclusions}\label{sec6}
We have developed both distributed continuous-time and discrete-time solvers for linear algebraic equations over a network, where the node-to-node  communication  was  scalars. The scheme contained first  a dimension compressing vector that  generates a real-valued  message for node communication as an inner product, and then a data unfolding step in the local computations where the scalar message is  unfolded  along the subspace generated by the compressing vector.  In view of such a simple idea, we proposed a consensus flow that can deliver linear convergence at a network level, while all node-to-node communications were kept as real numbers for all time. Built on such a compressed average consensus, distributed continuous-time and discrete-time solvers for linear algebraic equations were then proposed with linear convergence guarantees. In simulations, the compression vector was shown to be able to  reduce the communication complexity of the overall computation process.  {
What this work left open for future study is whether the scalarized compressor ensures the effectiveness of the algorithm when solving distributed convex and non-convex optimization problems. 
The issue of communication loss is also worth study, particularly when stochastic loss occurs. Notably, we will try to combine scalarized compressors with other techniques,  such as event-triggered methods,
in order to reduce the number of communication bits per communication while  reducing the number of communications.}

\appendix

\subsection{A useful lemma}

Let $x\in \mathbb{R}^b$ and $\Phi(\cdot)$ be a symmetric piecewise continuous and  bounded  matrix function mapping from $\mathbb{R}^{\geq 0}$ to $\mathbb{R}^{b\times b}$, and denote $\bar \Phi=\sup\limits_{t\geq t_0}\|\Phi(t)\|$. Consider  {LTV system}
\begin{equation}\label{LTV}
\dot{x}=-\Phi(t) x\,,
\end{equation}
for which we have the following result.

\begin{lemma}\label{lem1}
 The origin of system (\ref{LTV})  is  globally exponentially  stable (GES), i.e., there exist constants $k_x,\gamma_x>0$ such that, for all $x(t_0)\in\mathbb{R}^{b}$
 \[
 \|x(t)\|^2 \leq k_x\|x(t_0)\|^2 e^{-\gamma_x (t-t_0)}\,,\quad \forall t\geq t_0
 \]
iff the matrix square root of $\Phi$, i.e., $\Phi^{\frac{1}{2}}$ is PE, i.e., there exist  constants $\alpha_1,T_1>0$ such that
\begin{equation}\label{100}
\int_{t}^{t+T_1} \Phi (s) ds \geq \alpha_1 \Ib_b\,,\quad \forall t\geq t_0\,.
\end{equation}
Moreover, if (\ref{100}) holds, then there holds
\begin{equation}\label{eq:app-gammak}\ba{l}
\gamma_x=-\frac{1}{T_1}\ln\Big(1-\frac{2\alpha_1}{(1+\bar\Phi T_1)^2}\Big)\,,
k_x={1}/\Big(1-\frac{2\alpha_1}{(1+\bar\Phi T_1)^2}\Big)
\ea\,.
\end{equation}
\end{lemma}
\begin{proof}
The first part of the statement on the equivalence between the GES of (\ref{LTV}) and the PE condition  \eqref{100} is clear by \cite{Brian-TAC-1977,Narendra-SIAM-1977}. Regarding the proof on the convergence rate of the  system (\ref{LTV}) under the PE condition, it follows along the same line of the proof of \cite[Lemma 5]{Loria(2002)}. For simplicity, a sketch of proof is given  below. 

By defining $V(t)={\|x(t)\|^2}/{2}$, one can follow the proof of \cite[Lemma 5]{Loria(2002)}  to show that 
\[
\frac{\alpha_1\rho}{1+\rho}\|x(t)\|^2\leq (1+\rho\bar\Phi^2T_1^2)[V(t)-V(t+T_1)]
\]
holds for all $\rho>0$, implying
\[
V(t+T_1) \leq \Big(1-\frac{2\alpha_1}{(1+\bar\Phi T_1)^2}\Big)V(t)\,
\]
by fixing $\rho=\frac{1}{\bar\Phi T_1}$.
As a result, it follows that
\[
V(t) \leq k_x V(t_0) e^{-\gamma_x(t-t_0)}\,,\quad \forall t\geq t_0
\]
with $k_x,\gamma_x$ given in \eqref{eq:app-gammak}. This, by $V(t)={\|x(t)\|^2}/{2}$, completes the proof.
\end{proof}

\subsection{Proof of Theorem \ref{thm1}}
\label{app:proof-thm1}

By collecting all agent flows (\ref{ccf}), we can write the network flow into the following compact form
\begin{equation}\label{laplacian}
\dot{\mathbf{x}}(t)  =- \Big( \mathbf{L}_{\mathrm{G}}\otimes \mathbf{C}(t)\mathbf{C}^\top(t)\Big) \mathbf{x}(t).
\end{equation}
In the following, we denote
\begin{equation}
\label{eq:Ccal}
    \mathcal{C}(t)= \mathbf{C}(t)\mathbf{C}^\top(t)\,,
\end{equation}
and employ Lemma \ref{lem1} to prove Theorem \ref{thm1}. 
The key idea lies in showing that the GLC of (\ref{laplacian}) is equivalent to the GES of a  {LTV system} in the form \eqref{LTV}, for which the PE condition \eqref{100} holds if and only if Assumption \ref{ass-PE} holds. This completes the proof by recalling the first part of Lemma \ref{lem1}.

In view of the above intuition,  we first transform the GLC of (\ref{ccf}) to the GES of a LTV system.
As $\mathbf{1}_n^\top\mathbf{L}_{\mathrm{G}}=0$, we have
$\frac{1}{n}(\mathbf{1}_n^\top\otimes \Ib_m)\dot{\mathbf{x}}(t) = 0
$ for all $t\geq t_0$,
which implies 
\begin{equation}\label{eq:appb-xbstar}\ba{l}
\xb^\star = \frac{1}{n}(\mathbf{1}_n^\top\otimes \Ib_m){\mathbf{x}}(t_0)= \frac{1}{n}(\mathbf{1}_n^\top\otimes \Ib_m){\mathbf{x}}(t)\,.
\ea\end{equation}
With this in mind, we let $S\in\mathbb{R}^{n\times(n-1)}$ be a matrix whose rows being eigenvalue vectors corresponding to nonzero eigenvalues of $\mathbf{L}_{\mathrm{G}}$,  satisfying
\begin{equation}\label{eq-S}\ba{l}
    S^\top\mathbf{1}_n =0\,,\quad  \mathbf{I}_n = SS^\top +\mathbf{1}_n\mathbf{1}_n^\top/n\,, \\ S^\top S = \mathbf{I}_{n-1}\,,\quad S^\top\mathbf{L}_{\mathrm{G}}S=\Lambda:=\diag\{\lambda_2,\ldots,\lambda_n\}\,.
\ea\end{equation}
Define $\Sb := S\otimes \mathbf{I}_m$ and 
\begin{equation}\label{eq:appc-y}
\mathbf{z}:=\Sb^\top\mathbf{x}\,,
\end{equation}
whose  time derivative  along \eqref{laplacian} is given by
\beeq{\label{ysys}\ba{l}
\dot{\mathbf{z}}(t)
=-\Big(S^\top\mathbf{L}_{\mathrm{G}}\otimes  \mathcal{C}(t)\Big) \mathbf{x}(t)\\
= -\Big(S^\top\mathbf{L}_{\mathrm{G}}(SS^\top  +\mathbf{1}_n\mathbf{1}_n^\top/n)\otimes  \mathcal{C}(t)\Big) \mathbf{x}(t)\\
= -\Big(\Lambda \otimes  \mathcal{C}(t)\Big) \mathbf{z}(t)
\ea}
where the second and third equalities are obtained by using the second equality of \eqref{eq-S} and the fact that $\mathbf{L}_{\mathrm{G}}\mathbf{1}_n=0$, respectively. As $\mathrm{G}$ is undirected and connected, by recalling \cite{magnusbook}, we have
\begin{equation}\label{eq_lambda2}
   \lambda_n \mathbf{I}_{n-1}  {\succeq}  \Lambda  {\succeq} \lambda_2 \mathbf{I}_{n-1} {\succ}0\,,
\end{equation}
with $\lambda_2, \lambda_n>0$ being the second smallest eigenvalue and the maximal eigenvalue of $\mathbf{L}_{\mathrm{G}}$, respectively. Note that the resulting $\zb$-flow is  {LTV system}  of the form \eqref{LTV} with a symmetric piecewise continuous and  bounded matrix $\Phi(t)=\Big(\Lambda \otimes  \mathcal{C}(t)\Big)$.

To this end,  we observe from \eqref{eq:appb-xbstar}, \eqref{eq-S} and \eqref{eq:appc-y} that
\begin{equation}\label{eq:appb-y}
\mathbf{x}(t)-\mathbf{1}_n\otimes\xb^\star = \Sb\mathbf{z}(t)
\end{equation}
which,  by Definition 1, proves that the flow (\ref{ccf}) achieves the GES (i.e., the inequality \eqref{eq:GEAC} holds),  \emph{if and only if} the linear time-varying system  of $\mathbf{z}(t)$ is GES at the origin.

It is further noted from Lemma \ref{lem1} that  the $\zb$-system  is GES at the origin if and only if  for some positive $\alpha_1$ and $T_1$,
\begin{equation}\label{eq:V-PE}
\int_{t}^{t+T_1} \Big(\Lambda \otimes \mathcal{C}(\tau)\Big) d\tau  {\succeq} \alpha_1\mathbf{I}_{mn}\,, \quad \forall t\geq t_0\,.
\end{equation}
Thus, the proof reduces to show that \eqref{eq:V-PE} holds \emph{if and only if} Assumption \ref{ass-PE} holds.

Note that
\[
  \int_{t}^{t+T}  \Big(\Lambda \otimes \mathcal{C}(\tau)\Big)  d\tau 
  =\Lambda \otimes\Big(\int_{t}^{t+T} \mathcal{C}(\tau) d\tau\Big)\,.
\]
By recalling \eqref{eq_lambda2}, it then can be easily verified that 
\[
   \mbox{Assumption } \ref{ass-PE} \quad \Longrightarrow\quad \int_{t}^{t+T} \psib(\tau)\psib^\top(\tau)  d\tau  {\succeq} \alpha\lambda_2 \mathbf{I}_{mn}\,,
\]
verifying \eqref{eq:V-PE}, and 
\[
\eqref{eq:V-PE} \quad \Longrightarrow\quad \int_{t}^{t+T_1} \mathbf{C}(\tau)\mathbf{C}^\top(\tau) d\tau   {\succeq} (\alpha_1/\lambda_n)\mathbf{I}_m\,,
\]
verifying Assumption \ref{ass-PE}. The proof  thus is  completed.

\subsection{Proof of Theorem \ref{thm2}}
\label{app:proof-thm2}

In what follows, we follow the arguments  in Appendix \ref{app:proof-thm1} to prove Theorem \ref{thm2}. Specifically, the proof lies in quantifying the  convergence rate of $\zb$-system \eqref{ysys} and then applying \eqref{eq:appb-y} to derive the convergence rate of the proposed flow (\ref{ccf}).

We first recall that $\dot \zb= -\Big(\Lambda\otimes  \mathcal{C}(t)\Big) \mathbf{z}$, and observe that
\[\ba{rcl}
\|\Lambda\otimes  \mathcal{C}(t)\|\leq \bar c^2 \lambda_n\,
\ea\]
with $S$ defined in \eqref{eq-S}.
Then, by recalling Lemma \ref{lem1} and using \eqref{eq:PE-bound} and $\|\mathcal{C}(t)\|=1$, we can obtain
\[
\|\mathbf{z}(t)\|^2 \leq c\|\mathbf{z}(t_0)\|^2 \gamma^{(t-t_0)}
\]
with $\gamma=\Big(1-\frac{2\alpha\lambda_2}{(1+T\lambda_n)^2}\Big)^{1/T}$, $
c={1}/\Big(1-\frac{2\alpha\lambda_2}{(1+T\lambda_n)^2}\Big)$.
The proof is thus completed by using $\|\mathbf{x}(t)-\mathbf{1}_n\otimes\xb^\star\| = \|\Sb\mathbf{z}(t)\| \leq \|\mathbf{z}(t)\|$ and $\|\mathbf{z}(t_0)\| = \|\Sb^\top\mathbf{x}(t_0)\|\leq \|\mathbf{x}(t_0)\|$.

\subsection{Proof of Theorem \ref{thm-LAE-new}}

\label{app:D}

The  network flow (\ref{eq:LAE_solver}) can be compactly described as
\begin{equation}
\dot{\mathbf{x}}(t)  = - \big( \mathbf{L}_{\mathrm{G}} \otimes \mathcal{C}(t)\big) \mathbf{x}(t) - s(\mathsf{H}\mathbf{x}(t)-\mathsf{B}),
\end{equation}
where
$\mathsf{H} = \mbox{blkdiag}\bigg({\Hb_1\Hb_1^\top},\cdots,{\Hb_n\Hb_n^\top}\bigg)
$ and $
\mathsf{B} = \big({b_1\Hb_1^\top},\cdots,{b_n\Hb_n^\top}\big)^\top.
$

Following the arguments and terminologies in Appendix \ref{app:proof-thm1}, we define
$
\mathbf{z}:=\Sb^\top\mathbf{x}\,
$
with $S$ satisfying \eqref{eq-S}, and  the error between the averaged state and the optimal solution as
\begin{equation}\label{eq-tilde_x}
\tilde{\mathbf{x}}(t) := \big(\frac{1}{n}\mathbb{I}^\top\big)\mathbf{x}(t) - \vb^\ast\,,
\end{equation}
where  $\mathbb{I}:=\mathbf{1}_n\otimes \mathbf{I}_m$ and $\vb^\ast$ is the unique solution, satisfying
\begin{equation}\label{eq-y_ast}
    \mathsf{H}\mathbb{I}\mathbf{v}^\ast=\mathsf{B}\,.
\end{equation}
This, together with \eqref{eq-tilde_x}, implies
\begin{equation}\label{eq-B}
\mathsf{H}\mathbb{I}\tilde{\mathbf{x}}(t) = \mathsf{H}\mathbb{I}\,\mathbb{I}^\top\mathbf{x}(t) - \mathsf{B}\,.
\end{equation}

By using $
\mathbf{z}=\Sb^\top\mathbf{x}$, it is clear that $\xb = \Sb\zb+ \mathbf{1}_n\otimes\tilde\xb + \mathbf{1}_n\otimes\vb^\ast$, which implies
\begin{equation}\label{eq:F2}
\|\xb(t)-\mathbf{1}_n\otimes\vb^\ast\| \leq \|\mathbf{z}(t)\| + \sqrt{n}\|\tilde{\mathbf{x}}(t)\|\,.
\end{equation}
Thus, to prove the linear convergence of $\xb(t)-\mathbf{1}_n\otimes\vb^\ast$, we  turn to study the linear or exponential convergence  of $(\mathbf{z}(t),\tilde{\mathbf{x}}(t))$.

We now proceed to derive the flows of $\mathbf{z}(t)$ and $\tilde{\mathbf{x}}(t)$, which can be described as
\begin{equation}
\label{eq:z-t}
    \ba{rcl}
    \dot \zb &=& - (\Lambda \otimes  \mathcal{C}(t)) {\mathbf{z}} - s\Sb^\top(\mathsf{H}\mathbf{x}-\mathsf{B})\\
    &=& - (\Lambda \otimes  \mathcal{C}(t)) {\mathbf{z}} - s\Sb^\top\cdot\\&&\cdot\left[\mathsf{H}\left(\Sb\Sb^\top +\mathbb{I}\,\mathbb{I}^\top/n\right)\mathbf{x}-\mathsf{B}\right]\\
    &=& - (\Lambda\otimes  \mathcal{C}(t)+s\Sb^\top\mathsf{H}\Sb) {\mathbf{z}}  - s\Sb^\top \mathsf{H}\mathbb{I}\tilde{\mathbf{x}}
    \ea
\end{equation}
where the second and third equations are obtained by using 
$\Sb\Sb^\top +\mathbb{I}\mathbb{I}^\top/n =\mathbf{I}_{nm}$ and \eqref{eq-B}, respectively,
and
\begin{equation}
\label{eq:x-t}
    \ba{rcl}
    \dot {\tilde\xb} &=& - \frac{s}{n}\mathbb{I}^\top(\mathsf{H}\mathbf{x}-\mathsf{B})\\
    &=& - \frac{s}{n}\mathbb{I}^\top\left[\mathsf{H}\left(\Sb\Sb^\top +\mathbb{I}\,\mathbb{I}^\top/n\right)\mathbf{x}-\mathsf{B}\right]\\
    &=& - \frac{s}{n}\mathbb{I}^\top\mathsf{H}\Sb\zb - \frac{s}{n}\mathbb{I}^\top
    \mathsf{H}\mathbb{I}\tilde{\mathbf{x}}.
    \ea
\end{equation}

Further, denote $\zetab = \begin{bmatrix}
    \zb\cr  \sqrt{n} \tilde\xb
\end{bmatrix}$, which yields
\begin{equation}\label{eq-sys_zetab}
    \dot\zetab = -\mathcal{A}(t)\zetab
\end{equation}
where $\mathcal{A}(t)$ is a \emph{symmetric} \emph{piece-wise continuous}  matrix of the form
\[
\mathcal{A} = 
\begin{bmatrix}
     \Lambda \otimes  \mathcal{C}(t)+s\Sb^\top\mathsf{H}\Sb &  \frac{s}{\sqrt{n}}\Sb^\top\mathsf{H}\mathbb{I}\cr
     \frac{s}{\sqrt{n}}\mathbb{I}^\top\mathsf{H}\Sb &  \frac{s}{n}\mathbb{I}^\top
    \mathsf{H}\mathbb{I}
\end{bmatrix}\,.
\]
Clearly, $\mathcal{A}(t)$ is bounded, which can be easily verified using \eqref{eq-S}, \eqref{eq_lambda2} and $\|\mathcal{C}(t)\|= 1$ to satisfy
\begin{equation}\label{eq:A-up}
    \|\mathcal{A}(t)\| \leq \lambda_n+ 2s\|\mathsf{H}\| = \lambda_n+ 2sh_M^2\,,\quad \forall t\geq t_0.
\end{equation}

Next we proceed to show that $\mathcal{A}^{{1}/{2}}$ is PE. Specifically,  let
\[\ba{l}
\Psi(t):=\int_{t}^{t+T}\mathcal{A}(\tau) d\tau \\
= 
\begin{bmatrix}
     \Lambda \otimes  \int_{t}^{t+T}\mathcal{C}(\tau)d\tau+T s\Sb^\top\mathsf{H}\Sb &  \frac{Ts}{\sqrt{n}}\Sb^\top\mathsf{H}\mathbb{I}\cr
     \frac{Ts}{\sqrt{n}}\mathbb{I}^\top\mathsf{H}\Sb &  \frac{Ts}{n}\mathbb{I}^\top
    \mathsf{H}\mathbb{I}
\end{bmatrix}
\ea\]
and show its positive-definite property by considering 
\[
\ba{l}
\begin{bmatrix}
    z\cr x
\end{bmatrix} ^\top\Psi(t) \begin{bmatrix}
    z\cr x
\end{bmatrix}
\geq \lambda_2\alpha \|z\|^2 + Ts \|\mathsf{H}^{1/2}\Sb z\|^2 \\  + 2\frac{Ts}{\sqrt{n}}z^\top\Sb^\top\mathsf{H}\mathbb{I} x + \frac{Ts}{n}\|\mathsf{H}^{1/2}\mathbb{I} x\|^2\\
\geq \lambda_2\alpha \|z\|^2 - \frac{\mu Ts}{1-\mu} \|\mathsf{H}^{1/2}\Sb z\|^2 +   \frac{\mu Ts}{n}\|\mathsf{H}^{1/2}\mathbb{I} x\|^2 \\
\geq \left(\lambda_2\alpha- \frac{\mu Ts}{1-\mu} h_M^2 \right) \|z\|^2 + \mu Ts \rho_m\|x\|^2
\ea
\]
for all $\frac{\lambda_2\alpha}{\lambda_2\alpha + Ts h_M^2}<\mu < 1$, and $z\in\mathbb{R}^{m(n-1)}$ and $x\in\mathbb{R}^{m}$. Moreover, by fixing
\[
\mu = \frac{\alpha' - \sqrt{\alpha'^2-4\lambda_2\alpha\rho_m Ts}}{2\rho_m T s}
\]
where $\alpha'=\lambda_2\alpha + (h_M^2+\rho_m)Ts$, it can be concluded that
\[
\Psi(t):=\int_{t}^{t+T}\mathcal{A}(\tau) d\tau  {\succeq} \bar\alpha \mathbf{I}_{mn}\,,\quad \forall t\geq t_0
\]
with
\[
\bar\alpha := \frac{\alpha' - \sqrt{\alpha'^2-4\lambda_2\alpha\rho_m Ts}}{2}\,.
\]

This in turn verifies that  $\mathcal{A}^{\frac{1}{2}}$ is PE, verifying the globally exponential stability of $\zetab$-system \eqref{eq-sys_zetab}, by recalling Lemma \ref{lem1} and using \eqref{eq:A-up}, with the convergence rate 
\[
\bar\gamma_f=-\frac{1}{T}\ln\Big(1-\frac{2\bar\alpha}{(1+(\lambda_n+ 2sh_M^2)) T)^2}\Big)\,.
\]
The proof is thus completed by noting $\gamma_f=e^{-\bar\gamma_f}$.

\subsection{Proof of Theorem \ref{thm-dt-DCO-new}}

\label{app:proof-thm-dt-dco-new}

In the following, we follow  terminologies of $\mathsf{H},\mathsf{B}$ in Appendix \ref{app:D}, and write
the overall discrete-time  algorithm (\ref{eq:dt-DCO}) in the following compact form as
\begin{equation}
{\mathbf{x}}[k+1]  = [\Ib_{mn} - h \big( \mathbf{L}_{\mathrm{G}} \otimes \mathcal{C}[k]\big) ]\mathbf{x}[k] - s(\mathsf{H}\mathbf{x}[k]-\mathsf{B})
\end{equation}

As a discrete-time counterpart of \eqref{eq:F2}, we have
\begin{equation}\label{eq:dt-F2}
\|\xb[k]-\mathbf{1}_n\otimes\vb^\ast\| \leq \|\mathbf{z}[k]\| + \sqrt{n}\|\tilde{\mathbf{x}}[k]\|\,.
\end{equation}
Thus, as in Appendix \ref{app:D}, we now proceed to study the discrete-time $(\mathbf{z},\tilde{\mathbf{x}})$-dynamics, which, by some straightforward calculations as in deriving \eqref{eq:z-t} and \eqref{eq:x-t}, takes the form
\begin{equation}
\label{eq:dt-z-tx}
    \ba{l}
    \zb[k+1] =  \left[\Ab[k]-s\Sb^\top\mathsf{H}\Sb\right] {\mathbf{z}}[k] - s\Sb^\top\mathsf{H}\mathbb{I}\tilde{\mathbf{x}}[k]\\
    {\tilde\xb}[k+1]= \left[\Ib_{n}- \frac{s}{n}\mathbb{I}^\top
    \mathsf{H}\mathbb{I}\right]\tilde\xb[k]- \frac{s}{n}\mathbb{I}^\top\mathsf{H}\Sb\zb[k] 
    \ea
\end{equation}
where  $\Ab[k]:= \Ib_{m(n-1)}- h\Lambda \otimes  \mathcal{C}[k]$ and $0<h<2/\lambda_n$. 

To analyze the $\zb$-subsystem in \eqref{eq:dt-z-tx}, we first consider the  system
\begin{equation}\label{eq:sys-theta}\ba{l}
    \theta[k+1] = \Ab[k] \theta[k]\\
    y[k] = \left((2h\Lambda-h^2\Lambda^2)^{1/2}\otimes\mathbf{C}^\top[k] \right)\theta[k]
\ea\end{equation}
for which the observability grammian is given by
\[
\mathcal{G}_K[k] = \sum_{j=0}^{K-1} \mathcal{T}^\top[k+j,k] \left((2h\Lambda-h^2\Lambda^2)\otimes\mathcal{C}[k] \right)\mathcal{T}[k+j,k]
\]
where $\mathcal{T}[k+j,k]$ 
the state transmission matrix, satisfying
\begin{equation}\label{eq:T_j}
    \mathcal{T}[k+j,k] = \prod_{i=0}^{j-1} \Ab[k+i]  \,.
\end{equation}
Clearly, $\|\mathcal{T}[k+j,k]\| \leq 1$ for all $0<h<2/\lambda_n$ and $j\geq0$.
Then by letting $V_\theta[k]= \|\theta[k]\|^2$, it immediately follows that
\[\ba{l}
V_\theta[k+1] - V_\theta[k] = -\theta^\top[k]\left((2h\Lambda-h^2\Lambda^2)\otimes\mathcal{C}[k] \right)\theta[k]\\
V_\theta[k+2] - V_\theta[k] = -\theta^\top[k]\left((2h\Lambda-h^2\Lambda^2)\otimes\mathcal{C}[k] \right)\theta[k]\\
-\theta^\top[k]\mathcal{T}^\top[k+1,k]\left((2h\Lambda-h^2\Lambda^2)\otimes\mathcal{C}[k] \right)\mathcal{T}[k+2,k]\theta[k]\,
\ea\]
which further implies
\[
V_\theta[k+K] - V_\theta[k] = -\theta^\top[k] \mathcal{G}_K[k] \theta[k]\,.
\]

Thus, we can conclude that the system \eqref{eq:sys-theta} is globally exponentially convergent to zero, if there exists $1>g>0$ such that
\begin{equation}\label{eq:g}
    \mathcal{G}_K[k]  {\preceq} g\Ib_{m(n-1)} \,
\end{equation}
or, equivalently
\begin{equation} \label{eq:g-2}
   \mathcal{T}^\top[k+K,k] \mathcal{T}[k+K,k]  {\preceq} (1-g)\Ib_{m(n-1)}\,. 
\end{equation}

To show \eqref{eq:g}, we then consider the following auxiliary system
\begin{equation}\label{eq:sys-theta-2}\ba{l}
    \theta'[k+1] = \Ab[k] \theta'[k] + F[k]y'[k] \\
    y'[k] = \left((2h\Lambda-h^2\Lambda^2)^{1/2}\otimes\mathbf{C}^\top[k] \right)\theta'[k]
\ea\end{equation}
with  $F[k] = h\Lambda(2h\Lambda-h^2\Lambda^2)^{-1/2}\otimes {\mathbf{C}[k]}$. For such auxiliary system,  it can be easily verified that the state transmission matrix is the identity matrix and the observability grammian is given by
\[\ba{rcl}
\mathcal{G}_K'[k] &=& \sum_{j=0}^{K-1}\left((2h\Lambda-h^2\Lambda^2)\otimes\mathcal{C}[k+j] \right)\\
&=& (2h\Lambda-h^2\Lambda^2)\otimes \sum_{j=0}^{K-1}\mathcal{C}[k+j] \\
& {\succeq}& \bar\alpha_d\Ib_{m(n-1)} {\succ}0
\ea\]
where we defined
\[
\bar\alpha_d := \alpha_d\min\{2h\lambda_2-h^2\lambda_2^2,2h\lambda_n-h^2\lambda_n^2\}\,.
\]

It is worth noting that the  auxiliary system \eqref{eq:sys-theta-2} is obtained from \eqref{eq:sys-theta} by introducing a bounded output feedback $F[k]y'[k]$. Thus, by recalling \cite[page 217]{Kailath(1980)} where the observability is shown to be invariant under bounded output feedback, we can conclude that \eqref{eq:g} and thus \eqref{eq:g-2} are proved.

With this in mind, we proceed to analyze the  $\zb$-subsystem in \eqref{eq:dt-z-tx}, and define the discrete-time Lyapunov function as
\begin{equation} \label{eq:range_v1}
   \|\zb[k]\|^2 \leq V_1[k]=\sum_{j=0}^{K-1}\|\mathcal{T}[k+j,k]\zb[k]\|^2 \leq K\|\zb[k]\|^2\,
\end{equation}
and along the $\zb$-subsystem in \eqref{eq:dt-z-tx} we define $\mathcal{T}^{k+j}_k:=\mathcal{T}[k+j,k]$, then we have
\begin{equation}\label{eq:dif_V1}
\ba{rcl}
&&V_1[k+1] - V_1[k] \\
&\leq& \sum_{j=1}^{K}\|\mathcal{T}^{k+j}_{k+1}\Ab[k]\zb[k]\|^2 - \sum_{j=0}^{K-1}\|\mathcal{T}^{k+j}_{k}\zb[k]\|^2\\ 
&& + 2\sum_{j=1}^{K}s\|\mathcal{T}^{k+j}_{k+1}\|^2\|\Ab[k]\zb[k]\Sb^\top\mathsf{H}[\Sb {\mathbf{z}}[k]+\mathbb{I}\tilde{\mathbf{x}}[k]]\| \\
&& +\sum_{j=1}^{K}s^2\|\mathcal{T}^{k+j}_{k+1}\|^2\|\Sb^\top\mathsf{H}[\Sb {\mathbf{z}}[k]+\mathbb{I}\tilde{\mathbf{x}}[k]]\|^2\\
&\leq & -g\|\zb[k]\|^2 
+2s K\|\zb[k]\mathsf{H}\| (\|\zb[k]\|+\sqrt{n}\|\tilde\xb[k]\|) \\ && +Ks^2 \|\mathsf{H}\|^2(\|\zb[k]\|^2 + 2\sqrt{n}\|\zb[k]\|\|\tilde\xb[k]\| + n\|\tilde\xb[k]\|^2)\\
&\leq& -\big(g-K(3sh_M^2+2s^2h_M^4)\big)\|\zb[k]\|^2  \\ &&+n K(sh_M^2+2s^2h_M^4)\|\tilde\xb[k]\|^2,\\
\ea\
\end{equation}
where the first inequality is obtained by substituting the upper equation of \eqref{eq:dt-z-tx},  the second is obtained by applying \eqref{eq:T_j}, \eqref{eq-S} and \eqref{eq:g-2}, and the last is obtained by  by the Young's Inequality.

For the $\tilde\xb$-subsystem in \eqref{eq:dt-z-tx}, we let $V_2[k]= \|\tilde\xb[k]\|^2$ and observe that
\begin{equation}\label{eq:dif_V2}
\ba{rcl}
&&V_2[k+1] - V_2[k] \\
&=& - 2s\tilde\xb[k]^\top\frac{1}{n}\mathbb{I}^\top\mathsf{H}\big[\mathbb{I}\tilde{\mathbf{x}}[k]+\Sb\zb[k]\big]\\ 
&& + s^2\|\frac{1}{n}\mathbb{I}^\top \mathsf{H}\big[\mathbb{I}\tilde{\mathbf{x}}[k]+\Sb\zb[k]\big]\|^2 \\
&\leq&  - 2s\rho_m \|\tilde\xb[k]\|^2 + s\frac{2h_M^2}{\sqrt{n}}\|\tilde\xb[k]\|\|\zb[k]\|\\ && +\frac{s^2h_M^4}{n}(\sqrt{n}\|\tilde\xb[k]\| + \|\zb[k]\| )^2\\
&\leq& -(s\rho_m-2s^2{h_M^4})\|\tilde\xb[k]\|^2 \\ &&+ (2s^2\frac{h_M^4}{n}+s\frac{h_M^4}{n\rho_m})\|\zb[k]\|^2,
\ea\
\end{equation}
where the last inequality is obtained by the Young's Inequality.

Thus, we let $V[k] = V_{1}[k] + p V_2[k]$ with $p=\frac{2{n}K h_M^2}{\rho_m}$, leading to
\[\ba{rcl}
&&V[k+1] - V[k] \\
&\leq& - \big(g-K(3sh_M^2+2s^2h_M^4+ \frac{4s^2h_M^6}{\rho_m} + \frac{2sh_M^6}{\rho_m^2})\big)\|\zb[k]\|^2\\ 
&& -p \Big(\frac{s\rho_m}{2} - s^2h_M^2(\rho_m+2{h_M^2})\Big)\|\tilde\xb[k]\|^2 \\
\ea\]
by using \eqref{eq:dif_V1}
 and \eqref{eq:dif_V2}.

With \eqref{eq:range_v1}, it can be easily found that for  $0< s < s^\ast$, with 
\[\ba{l}
s^\ast = \min\left\{\frac{\sqrt{(3+\frac{2h_m^4}{\rho_m^2})^2+\frac{4g}{K}(2+\frac{4 h_M^2}{\rho_m})}-(3+\frac{2h_M^4}{\rho_m^2})}{2h_M^2(2+4h_M^2/\rho_m)},\frac{\rho_m}{2h_M^2(\rho_m+2h_M^2)}\right\}
\ea\]
there holds
$
V[k+1] - V[k] \leq -\beta V[k]
$
with
\[\ba{l}
\beta := \min\left\{\Big(g-K(3sh_M^2+2s^2h_M^4+ \frac{4s^2h_M^6}{\rho_m} + \frac{2sh_M^6}{\rho_m^2})\Big),\right. \\ \left.\qquad\qquad \quad \Big(\frac{s\rho_m}{2} - s^2h_M^2(\rho_m+2{h_M^2})\Big)\right\}\,
\ea\]
 {
satisfying $\beta\in(0,1)$. 
This yields  $V[k]=\mathcal{O}({-\gamma_d^k})$, where $\gamma_d:=1-\beta$ satisfies $\gamma_d\in(0,1)$. With the definition of $V[k]$, we derive $\|\zb[k]\|^2=\mathcal{O}({-\gamma_d^k})$ and $\|\tilde\xb[k]\|^2=\mathcal{O}({-\gamma_d^k})$, which together with \eqref{eq:dt-F2} completes the proof.}


\end{document}